\documentclass[conference,letterpaper]{IEEEtran}
\usepackage{cite}
\usepackage{url}
\usepackage{ifthen}
\usepackage{algorithm, algorithmicx, algpseudocode}
\usepackage{graphicx}
\usepackage{textcomp}
\usepackage{booktabs}
\usepackage[cmex10]{amsmath}
\usepackage{amssymb,amsfonts}
\usepackage{pifont}
\usepackage[left=1.62cm,right=1.62cm,top=1.85cm]{geometry}
\usepackage{amsthm}
\usepackage{mathrsfs}
\def\BibTeX{{\rm B\kern-.05em{\sc i\kern-.025em b}\kern-.08em T\kern-.1667em\lower.7ex\hbox{E}\kern-.125emX}}
\usepackage{tikz}
\usetikzlibrary{automata,positioning,calc}
\usetikzlibrary{intersections}
\usetikzlibrary{decorations.pathreplacing,angles,quotes}

\usepackage{bm}
\usepackage{amscd}

\algnewcommand{\Initialize}[1]{%
  \State \textbf{Initialization:}
  \Statex \hspace*{\algorithmicindent}\parbox[t]{0.8\linewidth}{\raggedright #1}
}
\makeatletter

\theoremstyle{definition}
\newtheorem{theorem}{Theorem}
\newtheorem{definition}{Definition}
\newtheorem{lemma}{Lemma}
\newtheorem{prop}{Proposition}
\newtheorem{cor}{Corollary}
\newtheorem{remark}{Remark}

\newcommand{\abs}[1]{\left\lvert#1\right\rvert}

\newcommand{\argmax}{\operatornamewithlimits{argmax}}

\def\bn{\mathbb N}

\def\br{\mathbb R}

\def\vE{\mathbb E}

\font\b=cmr10 scaled\magstep4

\def\bigzerou{\smash{\lower1.7ex\hbox{\b 0}}}
\def\bigzerou{\smash{\lower1.7ex\hbox{\b 0}}}
\IEEEoverridecommandlockouts
\begin{document}

\title{New Algorithms for Computing Sibson Capacity and Arimoto Capacity
\thanks{This work was supported by JSPS KAKENHI Grant Number JP23K16886.}
}

\author{%
  \IEEEauthorblockN{Akira Kamatsuka\IEEEauthorrefmark{1},
                    Yuki Ishikawa\IEEEauthorrefmark{2},
                    Koki Kazama\IEEEauthorrefmark{1},
                    and Takahiro Yoshida\IEEEauthorrefmark{3}}
  \IEEEauthorblockA{\IEEEauthorrefmark{1}%
                   Shonan Institute of Technology, Japan, 
                    \{kamatsuka, kazama\}@info.shonan-it.ac.jp}
  \IEEEauthorblockA{\IEEEauthorrefmark{2}%
                    The University of Electro-Communications, Tokyo, Japan, 
                    i2431017@edu.cc.uec.ac.jp}
  \IEEEauthorblockA{\IEEEauthorrefmark{3}%
                    Nihon University, Japan, 
                    yoshida.takahiro@nihon-u.ac.jp}
}

\maketitle

\begin{abstract}
The Sibson and Arimoto capacity, which are  based on the Sibson and Arimoto mutual information (MI) of order $\alpha$, respectively, 
are well-known generalizations of the channel capacity $C$. 
In this study, we derive novel alternating optimization algorithms for computing these capacities by providing new variational characterizations of the Sibson  and Arimoto MI. 
Moreover, we prove that all iterative algorithms for computing these capacities are equivalent under appropriate conditions imposed on their initial distributions. 
\end{abstract}

\section{Introduction} \label{sec:intro}
Channel capacity $C:= \max_{p_{X}}I(X; Y)$, where $p_{X}$ is an input distribution and $I(X; Y)$ is the Shannon mutual information (MI), is a fundamental quantity in information theory because it characterizes the supremum of the achievable rate in a channel coding problem of a discrete memoryless channel $p_{Y\mid X}$ \cite{shannon}. 
In literature, there are several ways to generalize the channel capacity. 

A well-known generalization of the capacity is a class of \textit{capacity of order $\alpha$} (\textit{$\alpha$-capacity}) $C_{\alpha}^{(\cdot)}:=\max_{p_{X}}I_{\alpha}^{(\cdot)}(X; Y), \alpha\in (0, 1)\cup (1, \infty)$, 
where $I_{\alpha}^{(\cdot)}(X; Y)$ is a variant of MI (referred to as $\alpha$-MI \cite{7308959}) such as Sibson MI $I_{\alpha}^{\text{S}}(X; Y)$ \cite{Sibson1969InformationR}, 
Arimoto MI $I_{\alpha}^{\text{A}}(X; Y)$ \cite{arimoto1977}, 
and  Augustin--Csisz\'{a}r MI $I_{\alpha}^{\text{C}}(X; Y)$ \cite{370121},\cite{augusting_phd_thesis}. 
Recently, Liao \textit{et al.} reported the operational meaning of Arimoto and Sibson capacity in privacy-guaranteed data-publishing problems \cite{8804205}. 

The Arimoto--Blahut algorithm, which is a well-known alternating optimization algorithm, proposed and developed by 
Arimoto \cite{1054753}, Blahut \cite{1054855}, and other authors 
\cite{1405276}, \cite{5484972}, \cite{7035101},\cite{Toyota:2020aa},\cite{2641,9476038},\cite{nakagawa2022proof}
is used for calculating capacity $C$. 
To derive the algorithm, Arimoto and Blahut provided a variational characterization of Shannon MI, i.e., 
transformed its definition into an optimization problem.
Extending his results, Arimoto showed the equivalence between Arimoto capacity and Sibson capacity\footnote{Csisz\'{a}r showed that Augustin--Csisz\'{a}r capacity $C_{\alpha}^{\text{C}}$ is equal to Sibson capacity $C_{\alpha}^{\text{S}}$ \cite{370121} (see \cite{e22050526},\cite{Nakiboglu:2019aa}, and \cite{e23020199} for more general results).}, i.e., $C_{\alpha}^{\text{A}} = C_{\alpha}^{\text{S}}$,
and derived an alternating optimization algorithm for computing the Sibson capacity $C_{\alpha}^{\text{S}}$ by providing a variational characterization of Sibson MI \cite{arimoto1977,1055640}\footnote{Note that this is an interpretation of Arimoto's work from a current perspective because Sibson MI was not widely known at that time. Also note that computing Sibson capacity partly corresponds to computing the error exponent \cite{Gallager:1968:ITR:578869} and the correct decoding probability exponent \cite{1055007} because Sibson MI can be represented by the Gallager error exponent function $E_{0}(\rho, p_{X})$ \cite{Gallager:1968:ITR:578869}.}.
Later, Arimoto presented a similar iterative algorithm in his textbook (written in Japanese) \cite{BN01990060en} for directly calculating Arimoto capacity $C_{\alpha}^{\text{A}}$ 
by presenting a variational characterization of Arimoto MI without proof; however, the relation between these two iterative algorithms and the differences in their performance remain unclear. 

The main contribution of this paper are as follows:
\begin{itemize}
\item We describe the previously proposed algorithms \cite{arimoto1977,1055640,BN01990060en} in a unified manner in terms of \textit{$\alpha$-tilted distribution} \cite{8804205}. 
Using the $\alpha$-tilted distribution, we propose new algorithms for computing the Arimoto and Sibson capacities by presenting novel variational characterizations of Sibson and Arimoto MI (Theorem \ref{thm:novel_max_characterization} and \ref{thm:novel_update_formulae}). 
To derive the characterizations, we utilize the H{\"o}lrder's inequality.
\item We prove that all iterative algorithms are equivalent under appropriate conditions imposed on initial distributions (Theorem \ref{thm:equivalence_S1_S2_A1_A2} and Corollary \ref{cor:equivalence}).  
In Section \ref{sec:equivalence}, we provide a numerical example to demonstrate this equivalence. 
We also prove that all algorithms exhibit the global convergence property (Corollary \ref{cor:global_convergence}). 
\end{itemize}

\section{Preliminaries} \label{sec:preliminary}
Let $X$ and $Y$ be random variables on finite alphabets $\mathcal{X}$ and $\mathcal{Y}$, respectively.
Let $p_{X, Y} = p_{X}p_{Y\mid X}$ and $p_{Y}$ be a given joint distribution of $(X, Y)$ and a marginal distribution of $Y$, respectively. 
The set of all distributions $p_{X}$ is denoted as $\Delta_{\mathcal{X}}$. 
Let $H(X):=-\sum_{x}p_{X}(x)\log p_{X}(x)$, $H(X | Y):=-\sum_{x,y}p_{X}(x)p_{Y\mid X}(y | x)\log p_{X\mid Y}(x | y)$, 
and $I(X; Y) := H(X) - H(X | Y)$ be the Shannon entropy, conditional entropy, and Shannon MI, respectively.  
For a function of $X$, i.e., $f(X)$, we use $\vE_{X}[f(X)]$ to represent the expectation of $f(X)$. 
We also use $\vE_{X}^{p_{X}}[f(X)]$ to emphasize that we consider expectations in $p_{X}$. 
Throughout this study, we use $\log$ to represent the natural logarithm.

We initially review $\alpha$-MI, $\alpha$-capacity, and the Arimoto--Blahut algorithm.

\subsection{$\alpha$-mutual information and $\alpha$-capacity} \label{ssec:alpha_MI}

\begin{definition}
Let $\alpha\in (0, 1)\cup (1, \infty)$. Given distributions $p_{X}$ and $q_{X}$, the R{\' e}nyi entropy of order $\alpha$, 
denoted as $H_{\alpha}(p_{X}) = H_{\alpha}(X)$, and the R{\' e}nyi divergence between $p_{X}$ and $q_{X}$ of order $\alpha$, denoted as $D_{\alpha}(p_{X} || q_{X})$, 
are defined as follows:
\begin{align}
H_{\alpha}(X) &:= \frac{1}{1-\alpha} \log \sum_{x} p_{X}(x)^{\alpha}, \\
D_{\alpha}(p_{X} || q_{X}) &:= \frac{1}{\alpha-1}\log \sum_{x} p_{X}(x)^{\alpha}q_{X}(x)^{1-\alpha}.
\end{align}
\end{definition}

\begin{definition}
Let $\alpha\in (0, 1)\cup (1, \infty)$ and $(X, Y)\sim p_{X, Y}=p_{X}p_{Y\mid X}$. 
The Sibson MI of order $\alpha$, denoted as $I_{\alpha}^{\text{S}} (X; Y)$, and the Arimoto MI of order $\alpha$, denoted as $I_{\alpha}^{\text{A}} (X; Y)$, 
are defined as follows:
\begin{align}
I_{\alpha}^{\text{S}} (X; Y) &:= \min_{q_{Y}} D_{\alpha} (p_{X}p_{Y\mid X} || p_{X}q_{Y}) \\ 
&= \frac{\alpha}{\alpha-1}\log \sum_{y}\left( \sum_{x} p_{X}(x)p_{Y\mid X}(y\mid x)^{\alpha} \right)^{\frac{1}{\alpha}} \\ 
&=  \frac{\alpha}{1-\alpha} E_{0} \left( \frac{1}{\alpha}-1, p_{X} \right), \\ 
I_{\alpha}^{\text{A}}(X; Y) &:= H_{\alpha}(X) - H_{\alpha}^{\text{A}}(X\mid Y) \\
&= \frac{\alpha}{1-\alpha} E_{0}\left( \frac{1}{\alpha}-1, p_{X_{\alpha}} \right), 
\end{align}
where $E_{0}(\rho, p_{X}):= -\log \sum_{y}\left( \sum_{x}p_{X}(x)p_{Y\mid X}(y\mid x)^{\frac{1}{1+\rho}} \right)^{1+\rho}$ is the Gallager error exponent function 
\cite{Gallager:1968:ITR:578869}, 
$H_{\alpha}^{\text{A}}(X | Y):= \frac{\alpha}{1-\alpha}\log\sum_{y} \left( \sum_{x}p_{X}(x)^{\alpha}p_{Y\mid X}(y\mid x)^{\alpha} \right)^{\frac{1}{\alpha}}$ 
is the Arimoto conditional entropy of order $\alpha$ \cite{arimoto1977}, 
and $p_{X_{\alpha}}$ is the $\alpha$-tilted distribution \cite{8804205} (scaled distribution \cite{7308959}, escort distribution \cite{10.5555/3019383}) of $p_{X}$, defined as follows: 
\begin{align}
p_{X_{\alpha}}(x) &:= \frac{p_{X}(x)^{\alpha}}{\sum_{x}p_{X}(x)^{\alpha}}. \label{eq:alpha_tilted_dist}
\end{align} 
\end{definition}

\begin{remark}
Note that the values of $I_{\alpha}^{\text{S}} (X; Y)$ and $I_{\alpha}^{\text{A}} (X; Y)$ are extended by continuity to $\alpha=1$ and $\alpha=\infty$. 
For $\alpha=1$, $I_{\alpha}^{\text{S}} (X; Y)$ and $I_{\alpha}^{\text{A}} (X; Y)$ reduce to the Shannon MI, $I(X; Y)$.
\end{remark}

The $\alpha$-tilted distribution has the following properties: 
\begin{prop} \label{prop:properies_alpha_dist}
Let $\alpha, \beta\in (0, 1)\cup (1, \infty)$. Given a distribution $p_{X}$, 
\begin{enumerate}
\item the $\beta$-tilted distribution of the $\alpha$-tilted distribution of $p_{X}$ is the $\alpha\beta$-tilted distribution of $p_{X}$, i.e., 
$p_{(X_{\alpha})_{\beta}} = p_{X_{\alpha\beta}}$. Specifically, $p_{\left( X_{\alpha} \right)_{{1}/{\alpha}}} = p_{(X_{{1}/{\alpha}})_{\alpha}} = p_{X}$. 
\item Given a continuous function $\mathcal{F}\colon \Delta_{\mathcal{X}} \to \br$, 
\begin{align}
\max_{p_{X}} \mathcal{F}(p_{X}) = \max_{p_{X}} \mathcal{F}(p_{X_{\alpha}}), \label{eq:alpha_tilted_max}
\end{align}
where the maximum is taken over all distributions $p_{X}$.
\item Assume that $p_{X}$ has full support. Then, $p_{X} = p_{X_{\alpha}}$ if and only if $p_{X}$ is the uniform distribution on $\mathcal{X}$, i.e., $p_{X}(x) = 1/\abs{\mathcal{X}}, x\in \mathcal{X}$. 
\end{enumerate}
\end{prop}
\begin{proof}
See Appendix \ref{proof:properies_alpha_dist}. 
\end{proof}

Arimoto showed that using \eqref{eq:alpha_tilted_max} \cite{arimoto1977}, the Sibson and Arimoto capacities are equivalent.

\begin{definition}
Given a channel $p_{Y\mid X}$, Sibson capacity $C_{\alpha}^{\text{S}}$ and Arimoto capacity $C_{\alpha}^{\text{A}}$ are defined as follows, respectively:
\begin{align}
C_{\alpha}^{\text{S}} &= \max_{p_{X}} I_{\alpha}^{\text{S}}(X; Y), \\ 
C_{\alpha}^{\text{A}} &= \max_{p_{X}} I_{\alpha}^{\text{A}}(X; Y).
\end{align}
\end{definition}

\begin{prop}[\text{\cite[Lemma 1]{arimoto1977}}]
\begin{align}
C_{\alpha}^{\text{S}} &= C_{\alpha}^{\text{A}}.
\end{align}
\end{prop}

\begin{remark}
Csisz\'{a}r proposed another $\alpha$-MI (referred to as Augustin--Csisz\'{a}r MI of order $\alpha$), which is defined as $I_{\alpha}^{\text{C}}(X; Y):=\min_{q_{Y}}\vE_{X}\left[D_{\alpha}(p_{Y\mid X}(\cdot \mid X) || q_{Y})\right]$, and proved that Augustin--Csisz\'{a}r capacity $C_{\alpha}^{\text{C}}:=\max_{p_{X}}I_{\alpha}^{\text{C}}(X; Y)$ is equivalent to Sibson capacity $C_{\alpha}^{\text{S}}$ \cite[Prop. 1]{370121}.
Note that for $\alpha=1$, all these $\alpha$-capacities reduce to channel capacity $C:=\max_{p_{X}}I(X; Y)$.
\end{remark}

\subsection{Arimoto--Blahut algorithm}\label{ssec:ABA}
In this subsection, we review the well-known iterative algorithms proposed by Arimoto \cite{1054753, arimoto1977, 1055640, BN01990060en} and Blahut \cite{1054855} for computing channel capacity $C$ and $\alpha$-capacity $C_{\alpha}^{(\cdot)}$.
The key technique used to derive these iterative algorithms is the variational characterizations of MI on $q_{X\mid Y}=\{q_{X\mid Y}(\cdot\mid y)\}_{y\in \mathcal{Y}}$, where $q_{X\mid Y}(\cdot \mid y)$ is a conditional distribution of $X$, given $Y=y$. 

Arimoto \cite{1054753} and Blahut \cite{1054855} proved the following variational characterization of Shannon MI $I(X; Y)$. 

\begin{prop}[\text{\cite[Eqs.(10) and (11)]{1054753}}]\label{prop:max_characterization_Shannon_MI}
    \begin{align}
        I(X;Y) = \max_{q_{X\mid Y}} F(p_{X}, q_{X\mid Y}), \label{eq:max_characterization_Shannon_MI}
    \end{align}
    where $F(p_{X}, q_{X\mid Y}):= \vE_{X, Y}^{p_{X}p_{Y\mid X}}\left[\log \frac{q_{X\mid Y}(X\mid Y)}{p_{X}(X)}\right]$;  
    the maximum in \eqref{eq:max_characterization_Shannon_MI} is achieved at $q^{*}_{X\mid Y}(x | y) = p_{X\mid Y}(x | y):=\frac{ p_{X}(x)p_{Y\mid X}(y\mid x)}{\sum_{y}p_{X}(x)p_{Y\mid X}(y\mid x)}$. 
\end{prop}
%
Consequently, the channel capacity is represented as double maximum $C = \max_{p_{X}}\max_{q_{X\mid Y}} F(p_{X}, q_{X\mid Y})$ and 
the alternating optimization algorithm for computing $C$ can be derived as described in Algorithm \ref{alg:aba}, where $p_{X}^{(0)}$ is an initial distribution of the algorithm.
Figure \ref{fig:fig_aba} shows the iterations of $p_{X}^{(k)}$ and $q_{X\mid Y}^{(k)}$ in the algorithm.

\begin{algorithm}[h]
	\caption{Arimoto--Blahut Algorithm \cite{1054753},\cite{1054855}}
	\label{alg:aba}
	\begin{algorithmic}[1]
		\Require 
			\Statex $p_{X}^{(0)}, p_{Y\mid X}$, $\epsilon\in (0, 1)$
		\Ensure
			\Statex $C$
		\Initialize{
			$q_{X\mid Y}^{(0)} \gets \argmax_{q_{X\mid Y}}F(p_{X}^{(0)}, q_{X\mid Y})$\\ 
			$F^{(0, 0)}\gets F(p_{X}^{(0)}, q_{X\mid Y}^{(0)})$ \\ 
			$k\gets 0$ \\
      }
		\Repeat
			\State $k\gets k+1$
			\State $p_{X}^{(k)} \gets \argmax_{p_{X}}F(p_{X}, q_{X\mid Y}^{(k-1)})$

			\State $q_{X\mid Y}^{(k)} \gets \argmax_{q_{X\mid Y}}F(p_{X}^{(k)}, q_{X\mid Y})$
			\State $F^{(k, k)} \gets F(p_{X}^{(k)}, q_{X\mid Y}^{(k)})$
		\Until{$\abs{F^{(k, k)} - F^{(k-1, k-1)}} < \epsilon$} 
		\State \textbf{return} $F^{(k, k)}$
	\end{algorithmic}
\end{algorithm}

\begin{figure}[htbp]
\centering
\includegraphics[width=5cm, clip]{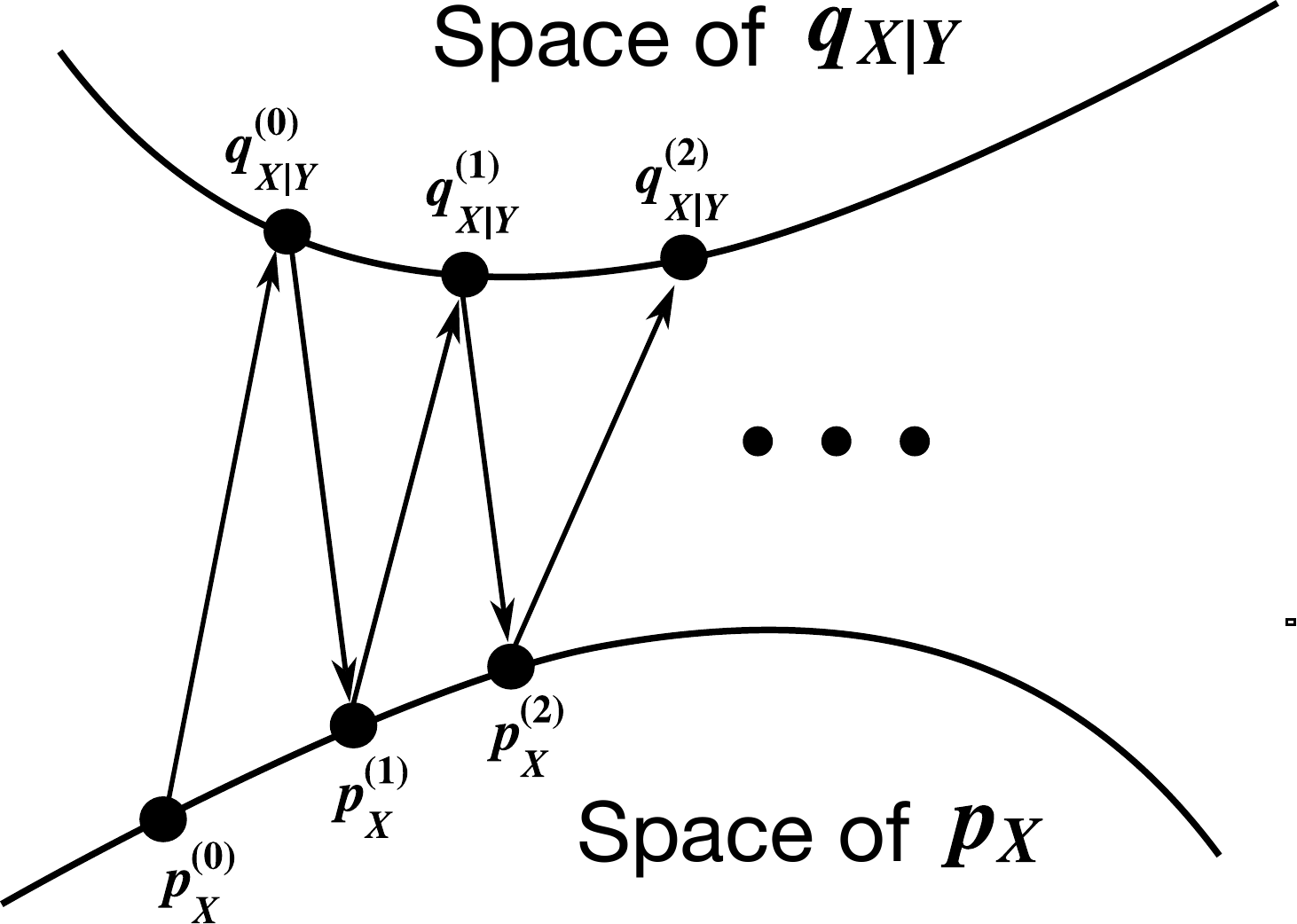}
\caption{Iterations of $p_{X}^{(k)}$ and $q_{X\mid Y}^{(k)}$ in the Arimoto--Blahut algorithm. }
\label{fig:fig_aba}
\end{figure}

Similarly, Arimoto \cite{arimoto1977, 1055640, BN01990060en} proposed the following variational characterizations 
of Sibson MI $I_{\alpha}^{\text{S}}(X; Y)$ and Arimoto MI $I_{\alpha}^{\text{A}}(X; Y)$.

\begin{prop}[\text{\cite[Thm 4]{arimoto1977},\cite[Thm 1]{1055640},\cite{BN01990060en}}] \label{prop:max_characterization_Sibson_Arimoto_MI}
\begin{align}
I_{\alpha}^{\text{S}}(X; Y) &= \max_{q_{X\mid Y}} F_{\alpha}^{\text{S}}(p_{X}, q_{X\mid Y}), \label{eq:max_characterization_Sibson_MI} \\ 
I_{\alpha}^{\text{A}}(X; Y) &= \max_{q_{X\mid Y}} F_{\alpha}^{\text{A}}(p_{X}, q_{X\mid Y}), \label{eq:max_characterization_Arimoto_MI}
\end{align}
where 
$F_{\alpha}^{\text{S}}(p_{X}, q_{X\mid Y}) := \frac{\alpha}{\alpha-1}\log \sum_{x, y}p_{X}(x)^{\frac{1}{\alpha}} p_{X\mid Y}(y | x) $\\ 
$q_{X\mid Y}(x | y)^{\frac{\alpha-1}{\alpha}}$, 
$F_{\alpha}^{\text{A}}(p_{X}, q_{X\mid Y}) := H_{\alpha}(p_{X}) - \frac{\alpha}{1-\alpha} \log \sum_{x, y}p_{X}(x)p_{Y\mid X}(y | x)q_{X\mid Y}(x | y)^{\frac{\alpha-1}{\alpha}}$; 
the maximum in \eqref{eq:max_characterization_Sibson_MI} is achieved at $q_{X\mid Y}^{*}(x | y):=\frac{p_{X}(x)p_{Y\mid X}(y\mid x)^{\alpha}}{\sum_{x}p_{X}(x)p_{Y\mid X}(y\mid x)^{\alpha}}$, and  
the maximum in \eqref{eq:max_characterization_Arimoto_MI} is achieved at $q_{X\mid Y}^{*}(x | y) := \frac{p_{X}(x)^{\alpha}p_{Y\mid X}(y\mid x)^{\alpha}}{\sum_{x}p_{X}(x)^{\alpha}p_{Y\mid X}(y\mid x)^{\alpha}}$. 
\end{prop}

In Algorithm \ref{alg:aba}, by replacing $F(p_{X}, q_{X\mid Y})$ with $F_{\alpha}^{\text{S}}(p_{X}, q_{X\mid Y})$ and $F_{\alpha}^{\text{A}}(p_{X}, q_{X\mid Y})$, 
iterative algorithms for calculating $C_{\alpha}^{\text{S}}=\max_{p_{X}}\max_{q_{X\mid Y}}F_{\alpha}^{\text{S}}(p_{X}, q_{X\mid Y})$ (Algorithm S1) and 
$C_{\alpha}^{\text{A}}=\max_{p_{X}}\max_{q_{X\mid Y}}F_{\alpha}^{\text{A}}(p_{X}, q_{X\mid Y})$ (Algorithm A1) are derived.
The update formulae of $p_{X}^{(k)}$ and $q_{X\mid Y}^{(k)}$ for each algorithm are presented in Table \ref{tab:update_formula}. 

\begin{table*}[h]
  \caption{Formulae for updating $p_{X}^{(k)}$ and $q_{X\mid Y}^{(k)}$ in the Arimoto--Blahut Algorithm for calculating $C, C_{\alpha}^{\text{S}}$ and $C_{\alpha}^{\text{A}}$}
  \label{tab:update_formula}
  \resizebox{1.\textwidth}{!}{
  \centering
  \begin{tabular}{@{} cccc @{}}
    \toprule
    Name & $F_{\alpha}^{(\cdot)}(p_{X}, q_{X\mid Y})$ & $p_{X}^{(k)}$ & $q_{X\mid Y}^{(k)}$  \\ 
    \midrule
    \begin{tabular}{c}
    Algorithm 1 for \\ 
    computing $C$ \cite{1054753}, \cite{1054855}
    \end{tabular}
    & \begin{tabular}{c}
    $\vE_{X, Y}^{p_{X}p_{Y\mid X}} \left[\log \frac{q_{X\mid Y}(X\mid Y)}{p_{X}(X)}\right]$ 
    \end{tabular}
    & $\frac{\prod_{y} q_{X\mid Y}^{(k-1)}(x\mid y)^{p_{Y\mid X}(y\mid x)}}{\sum_{x}\prod_{y} q_{X\mid Y}^{(k-1)}(x\mid y)^{p_{Y\mid X}(y\mid x)}}$ 
    & $\frac{p_{X}^{(k)}(x)p_{Y\mid X}(y\mid x)}{\sum_{x}p_{X}^{(k)}(x) p_{Y\mid X}(y\mid x)}$  \\ 
    \midrule
    \begin{tabular}{c}
    Algorithm S1 for \\
    computing $C_{\alpha}^{\text{S}}$ \cite{arimoto1977,1055640}
    \end{tabular}
    & \begin{tabular}{c}
    $\frac{\alpha}{\alpha-1}\log \sum_{x, y}p_{X}(x)^{\frac{1}{\alpha}} p_{Y\mid X}(y\mid x) q_{X\mid Y}(x\mid y)^{\frac{\alpha-1}{\alpha}}$ 
    \end{tabular}
    & $\frac{\left( \sum_{y}p_{Y\mid X}(y\mid x)q^{(k-1)}_{X\mid Y}(x\mid y)^{\frac{\alpha-1}{\alpha}} \right)^{\frac{\alpha}{\alpha-1}}}{\sum_{x} \left( \sum_{y}p_{Y\mid X}(y\mid x)q^{(k-1)}_{X\mid Y}(x\mid y)^{\frac{\alpha-1}{\alpha}} \right)^{\frac{\alpha}{\alpha-1}}}$ 
    & $\frac{p^{(k)}_{X}(x)p_{Y\mid X}(y\mid x)^{\alpha}}{\sum_{x}p^{(k)}_{X}(x)p_{Y\mid X}(y\mid x)^{\alpha}}$  \\ 
    \begin{tabular}{c}
    Algorithm S2 for \\
    computing $C_{\alpha}^{\text{S}}$ \\ 
    (\textbf{This work}) 
    \end{tabular}
    & \begin{tabular}{c}
    $\frac{\alpha}{\alpha-1} \log \sum_{x, y} p_{X}(x)^{\frac{1}{\alpha}}p_{Y\mid X}(y\mid x) q_{X_{\alpha}\mid Y}(x\mid y)^{\frac{\alpha-1}{\alpha}}$ 
    \end{tabular}
    & $\frac{\left( \sum_{y}p_{Y\mid X}(y\mid x) q_{X_{\alpha}\mid Y}^{(k-1)}(x\mid y)^{\frac{\alpha-1}{\alpha}} \right)^{\frac{\alpha}{\alpha-1}}}{\sum_{x}\left( \sum_{y}p_{Y\mid X}(y\mid x) q_{X_{\alpha}\mid Y}^{(k-1)}(x\mid y)^{\frac{\alpha-1}{\alpha}} \right)^{\frac{\alpha}{\alpha-1}}}$ 
    & $\frac{p_{X}^{(k)}(x)^{\frac{1}{\alpha}}p_{Y\mid X}(y\mid x)}{\sum_{x} p_{X}^{(k)}(x)^{\frac{1}{\alpha}}p_{Y\mid X}(y\mid x)}$  \\ 
    \midrule 
    \begin{tabular}{c}
    Algorithm A1 for \\
    computing $C_{\alpha}^{\text{A}}$ \cite{BN01990060en}
    \end{tabular}
    & \begin{tabular}{c}
    $\frac{\alpha}{\alpha-1} \log \sum_{x, y} p_{X_{\alpha}}(x)^{\frac{1}{\alpha}}p_{Y\mid X}(y\mid x) q_{X\mid Y}(x\mid y)^{\frac{\alpha-1}{\alpha}}$ 
    \end{tabular}
    & $\frac{\left( \sum_{y}p_{Y\mid X}(y\mid x)q_{X\mid Y}^{(k-1)}(x\mid y)^{\frac{\alpha-1}{\alpha}} \right)^{\frac{1}{\alpha-1}}}{\sum_{x}\left( \sum_{y}p_{Y\mid X}(y\mid x)q_{X\mid Y}^{(k-1)}(x\mid y)^{\frac{\alpha-1}{\alpha}} \right)^{\frac{1}{\alpha-1}}}$ 
    & $\frac{p_{X}^{(k)}(x)^{\alpha}p_{Y\mid X}(y\mid x)^{\alpha}}{\sum_{x}p_{X}^{(k)}(x)^{\alpha}p_{Y\mid X}(y\mid x)^{\alpha}}$  \\ 
    \begin{tabular}{c}
    Algorithm A2 for \\
    computing $C_{\alpha}^{\text{A}}$ \\ 
    (\textbf{This work}) 
    \end{tabular}
    & \begin{tabular}{c}
    $\frac{\alpha}{\alpha-1} \log \sum_{x, y} p_{X_{\alpha}}(x)^{\frac{1}{\alpha}}p_{Y\mid X}(y\mid x) q_{X_{\alpha}\mid Y}(x\mid y)^{\frac{\alpha-1}{\alpha}}$ 
    \end{tabular}
    & $\frac{\left( \sum_{y}p_{Y\mid X}(y\mid x) q_{X_{\alpha}\mid Y}^{(k-1)}(x\mid y)^{\frac{\alpha-1}{\alpha}} \right)^{\frac{1}{\alpha-1}}}{\sum_{x}\left( \sum_{y}p_{Y\mid X}(y\mid x) q_{X_{\alpha}\mid Y}^{(k-1)}(x\mid y)^{\frac{\alpha-1}{\alpha}} \right)^{\frac{1}{\alpha-1}}}$ 
    & $\frac{p_{X}^{(k)}(x)p_{Y\mid X}(y\mid x)}{\sum_{x}p_{X}^{(k)}(x)p_{Y\mid X}(y\mid x)}$  \\ 
    \bottomrule
  \end{tabular}
  }
\end{table*}

\begin{remark}
Although Arimoto provided the proof for \eqref{eq:max_characterization_Sibson_MI} using the Karush--Kuhn--Tucker (KKT) condition described in \cite{1055640}, 
he did not provide explicit proof for \eqref{eq:max_characterization_Arimoto_MI} described in \cite{BN01990060en}. 
In Appendix \ref{proof:novel_max_characterization}, we provide an alternative proof for Proposition \ref{prop:max_characterization_Sibson_Arimoto_MI} using H\"{o}lder's inequality. 
\end{remark}

\begin{remark} \label{remark:observation}
Note that a simple calculation shows that $F_{\alpha}^{\text{A}}(p_{X}, q_{X\mid Y}) = F_{\alpha}^{\text{S}}(p_{X_{\alpha}}, q_{X\mid Y})$.
\end{remark}

\section{New Algorithms for Calculating Sibson and Arimoto Capacities} \label{sec:new_ABA}
In this section, we propose new algorithms for calculating the Sibson and Arimoto capacities 
by presenting novel variational characterizations of Sibson and Arimoto MI. 
For this purpose, we employ the $\alpha$-tilted distribution and H\"{o}lder's inequality \cite{Holder1889}.

Based on Remark \ref{remark:observation}, we consider the following objective functions for alternating optimization algorithms; 
these functions are defined as follows:
\begin{align}
&\tilde{F}_{\alpha}^{S}(p_{X}, q_{X\mid Y}) := F_{\alpha}^{\text{S}}(p_{X}, q_{X_{\alpha}\mid Y}) \label{eq:tilde_F_S} \\ 
&= \frac{\alpha}{\alpha-1} \log \sum_{x, y} p_{X}(x)^{\frac{1}{\alpha}}p_{Y\mid X}(y\mid x) q_{X_{\alpha}\mid Y}(x\mid y)^{\frac{\alpha-1}{\alpha}},  \\ 
&\tilde{F}_{\alpha}^{\text{A}}(p_{X}, q_{X\mid Y}) := {F}_{\alpha}^{\text{A}}(p_{X}, q_{X_{\alpha}\mid Y}) = {F}_{\alpha}^{\text{S}}(p_{X_{\alpha}}, q_{X_{\alpha}\mid Y}) \label{eq:tilde_F_A} \\ 
&= \frac{\alpha}{\alpha-1} \log \sum_{x, y} p_{X_{\alpha}}(x)^{\frac{1}{\alpha}}p_{Y\mid X}(y\mid x) q_{X_{\alpha}\mid Y}(x\mid y)^{\frac{\alpha-1}{\alpha}}, 
\end{align}
where $q_{X_{\alpha}\mid Y}=\{q_{X_{\alpha}\mid Y}(\cdot \mid y)\}_{y\in \mathcal{Y}}$ is a set of the $\alpha$-tilted distribution of $q_{X\mid Y}(\cdot \mid y)$, which is  
defined as $q_{X_{\alpha}\mid Y}(x | y) := \frac{q_{X\mid Y}(x\mid y)^{\alpha}}{\sum_{x}q_{X\mid Y}(x\mid y)^{\alpha}}$. 
In the following, we provide novel variational characterizations of Sibson and Arimoto MI. 

\begin{theorem} \label{thm:novel_max_characterization}
\begin{align}
I_{\alpha}^{\text{S}}(X; Y) &= \max_{q_{X\mid Y}} \tilde{F}_{\alpha}^{\text{S}}(p_{X}, q_{X\mid Y}), \label{eq:novel_max_characterization_Sibson_MI} \\ 
I_{\alpha}^{\text{A}}(X; Y) &= \max_{q_{X\mid Y}} \tilde{F}_{\alpha}^{\text{A}}(p_{X}, q_{X\mid Y}), \label{eq:novel_max_characterization_Arimoto_MI}
\end{align}
where the maximum in \eqref{eq:novel_max_characterization_Sibson_MI} is achieved at $q_{X\mid Y}^{*}(x | y) := \frac{p_{X}(x)^{\frac{1}{\alpha}}p_{Y\mid X}(y | x)}{\sum_{y}p_{X}(x)^{\frac{1}{\alpha}}p_{Y\mid X}(y | x)}$, and the maximum in \eqref{eq:novel_max_characterization_Arimoto_MI} is achieved at $q_{X\mid Y}^{*}(x | y) := p_{X\mid Y}(x|y)=\frac{p_{X}(x)p_{Y\mid X}(y | x)}{\sum_{y}p_{X}(x)p_{Y\mid X}(y | x)}$. 
\end{theorem}
\begin{proof}
See Appendix \ref{proof:novel_max_characterization}. 
\end{proof}

\begin{theorem} \label{thm:novel_update_formulae}
\begin{enumerate}
\item For a fixed $p_{X}$, $\tilde{F}_{\alpha}^{\text{S}}(p_{X}, q_{X\mid Y})$ is maximized by 
\begin{align}
q_{X\mid Y}^{*}(x\mid y) &= \frac{p_{X}(x)^{\frac{1}{\alpha}}p_{Y\mid X}(y | x)}{\sum_{x}p_{X}(x)^{\frac{1}{\alpha}}p_{Y\mid X}(y | x)}.
\end{align}
\item For a fixed $q_{X\mid Y}$, $\tilde{F}_{\alpha}^{\text{S}}(p_{X}, q_{X\mid Y})$ is maximized by 
\begin{align}
p_{X}^{*}(x) 
&= \frac{\left( \sum_{y}p_{Y\mid X}(y\mid x) q_{X_{\alpha}\mid Y}(x\mid y)^{\frac{\alpha-1}{\alpha}} \right)^{\frac{\alpha}{\alpha-1}}}{\sum_{x}\left( \sum_{y}p_{Y\mid X}(y\mid x) q_{X_{\alpha}\mid Y}(x\mid y)^{\frac{\alpha-1}{\alpha}} \right)^{\frac{\alpha}{\alpha-1}}}.
\end{align}
\item For a fixed $p_{X}$, $\tilde{F}_{\alpha}^{\text{A}}(p_{X}, q_{X\mid Y})$ is maximized by 
\begin{align}
q_{X\mid Y}^{*}(x\mid y) &= \frac{p_{X}(x)p_{Y\mid X}(y | x)}{\sum_{x}p_{X}(x)p_{Y\mid X}(y | x)}.
\end{align}
\item For a fixed $q_{X\mid Y}$, $\tilde{F}_{\alpha}^{\text{A}}(p_{X}, q_{X\mid Y})$ is maximized by 
\begin{align}
p_{X}^{*}(x) 
&= \frac{\left( \sum_{y}p_{Y\mid X}(y\mid x)q_{X_{\alpha}\mid Y}(x\mid y)^{\frac{\alpha-1}{\alpha}} \right)^{\frac{1}{\alpha-1}}}{\sum_{x} \left( \sum_{y}p_{Y\mid X}(y\mid x)q_{X_{\alpha}\mid Y}(x\mid y)^{\frac{\alpha-1}{\alpha}} \right)^{\frac{1}{\alpha-1}}}.
\end{align}
\end{enumerate}
\end{theorem}
\begin{proof}
See Appendix \ref{proof:novel_update_formulae}. 
\end{proof}

Using Theorem \ref{thm:novel_max_characterization} and Theorem \ref{thm:novel_update_formulae}, 
Algorithm S2 and A2 are derived by replacing $F(p_{X}, q_{X\mid Y})$ in Algorithm \ref{alg:aba} 
with $\tilde{F}_{\alpha}^{\text{S}}(p_{X}, q_{X\mid Y})$ and $\tilde{F}_{\alpha}^{\text{A}}(p_{X}, q_{X\mid Y})$. 
The formulae used for updating $p_{X}^{(k)}$ and $q_{X\mid Y}^{(k)}$ in these algorithms are presented in Table \ref{tab:update_formula}.

\section{Equivalence of algorithms}\label{sec:equivalence}
In this section, we prove the equivalence of alternating optimization algorithms used for computing the Sibson and Arimoto capacities. 

\subsection{Equivalence of algorithms}

Here, we denote the objective functions of each iterative algorithm as 
$F_{\alpha}^{\text{S1}}(p_{X}, q_{X\mid Y}):=F_{\alpha}^{\text{S}}(p_{X}, q_{X\mid Y})$, 
$F_{\alpha}^{\text{S2}}(p_{X}, q_{X\mid Y}):=\tilde{F}_{\alpha}^{\text{S}}(p_{X}, q_{X\mid Y})$, 
$F_{\alpha}^{\text{A1}}(p_{X}, q_{X\mid Y}):=F_{\alpha}^{\text{A}}(p_{X}, q_{X\mid Y})$, and
$F_{\alpha}^{\text{A2}}(p_{X}, q_{X\mid Y}):=\tilde{F}_{\alpha}^{\text{A}}(p_{X}, q_{X\mid Y})$. 
Let $\{p_{X}^{(k), (\cdot)}\}_{k=0}^{\infty}$ and $\{q_{X\mid Y}^{(k), (\cdot)}\}_{k=0}^{\infty}$ be sequences of distributions obtained from each algorithm using the updating formulae presented in Table \ref{alg:aba}, where $p_{X}^{(0), (\cdot)}$ is the initial distribution of each algorithm. 
Let $\{F_{\alpha}^{(k, k), (\cdot)}\}_{k=0}^{\infty}$ and $\{F_{\alpha}^{(k+1, k), (\cdot)}\}_{k=0}^{\infty}$ be sequences of values of objective functions 
defined as $F_{\alpha}^{(k, k), (\cdot)} := F_{\alpha}^{(\cdot)}(p_{X}^{(k)}, q_{X\mid Y}^{(k)})$ and $F_{\alpha}^{(k+1, k), (\cdot)} := F_{\alpha}^{(\cdot)}(p_{X}^{(k+1)}, q_{X\mid Y}^{(k)})$.
We obtain the following equivalence results.

\begin{theorem} \label{thm:equivalence_S1_S2_A1_A2} 
Let $\alpha\in (0, 1)\cup (1, \infty)$. 
Suppose that $p_{X}^{(0), \text{S1}} = p_{X}^{(0), \text{S2}}$ and $p_{X}^{(0), \text{A1}} = p_{X}^{(0), \text{A2}}$. Then, 
\begin{enumerate}
\item For all $k\in \bn$, 
\begin{align}
p_{X}^{(k), \text{S1}} &= p_{X}^{(k), \text{S2}}, & q_{X\mid Y}^{(k), \text{S1}} &= q_{X_{\alpha}\mid Y}^{(k), \text{S2}}, \label{eq:equivalence_S1_S2_dist} \\ 
F_{\alpha}^{(k, k), \text{S1}} &= F_{\alpha}^{(k, k), \text{S2}}, & F_{\alpha}^{(k+1, k), \text{S1}} &= F_{\alpha}^{(k+1, k), \text{S2}}. \label{eq:equivalence_S1_S2_value}
\end{align}
\item For all $k\in \bn$, 
\begin{align}
p_{X}^{(k), \text{A1}} &= p_{X}^{(k), \text{A2}}, & q_{X\mid Y}^{(k), \text{A1}} &= q_{X_{\alpha}\mid Y}^{(k), \text{A2}}, \label{eq:equivalence_A1_A2_dist} \\ 
F_{\alpha}^{(k, k), \text{A1}} &= F_{\alpha}^{(k, k), \text{A2}}, & F_{\alpha}^{(k+1, k), \text{A1}} &= F_{\alpha}^{(k+1, k), \text{A2}}. \label{eq:equivalence_A1_A2_value}
\end{align}
\end{enumerate}
Similarly, suppose that $p_{X}^{(0), \text{S1}} = p_{X_{\alpha}}^{(0), \text{A1}}$ and $p_{X}^{(0), \text{S2}} = p_{X_{\alpha}}^{(0), \text{A2}}$. Then, 
\begin{enumerate}
\item[3)] For all $k\in \bn$, 
\begin{align}
p_{X}^{(k), \text{S1}} &= p_{X_{\alpha}}^{(k), \text{A1}}, & q_{X\mid Y}^{(k), \text{S1}} &= q_{X\mid Y}^{(k), \text{A1}}, \label{eq:equivalence_S1_A1_dist} \\ 
F_{\alpha}^{(k, k), \text{S1}} &= F_{\alpha}^{(k, k), \text{A1}}, & F_{\alpha}^{(k+1, k), \text{S1}} &= F_{\alpha}^{(k+1, k), \text{A1}}. \label{eq:equivalence_S1_A1_value}
\end{align}
\item[4)] For all $k\in \bn$, 
\begin{align}
p_{X}^{(k), \text{S2}} &= p_{X_{\alpha}}^{(k), \text{A2}}, & q_{X\mid Y}^{(k), \text{S2}} &= q_{X\mid Y}^{(k), \text{A2}}, \label{eq:equivalence_S2_A2_dist}\\ 
F_{\alpha}^{(k, k), \text{S2}} &= F_{\alpha}^{(k, k), \text{A2}}, & F_{\alpha}^{(k+1, k), \text{S2}} &= F_{\alpha}^{(k+1, k), \text{A2}}. \label{eq:equivalence_S2_A2_value}
\end{align}
\end{enumerate}
\end{theorem}
\begin{proof}
See Appendix \ref{proof:equivalence_S1_S2_A1_A2}. 
\end{proof}

\begin{remark}
The above results show that 
\begin{itemize}
\item Algorithms S1 and S2 (\textit{resp.}, Algorithms A1 and A2) are equivalent if we select initial distributions such that $p_{X}^{(0), \text{S1}} = p_{X}^{(0), \text{S2}}$ (\textit{resp.}, $p_{X}^{(0), \text{A1}} = p_{X}^{(0), \text{A2}}$).
\item Algorithms S1 and A1 (\textit{resp.}, Algorithms S2 and A2) are equivalent if we select initial distributions such that $p_{X}^{(0), \text{S1}} = p_{X_{\alpha}}^{(0), \text{A1}}$ (\textit{resp.}, $p_{X}^{(0), \text{S2}} = p_{X_{\alpha}}^{(0), \text{A2}}$).
\end{itemize}

\end{remark}
Figures \ref{fig:S1_S2} and \ref{fig:S1_A1} visualize the statements described in \eqref{eq:equivalence_S1_S2_dist} and \eqref{eq:equivalence_S1_A1_dist}, respectively.

\begin{figure}[htbp]
\centering
\includegraphics[width=5cm, clip]{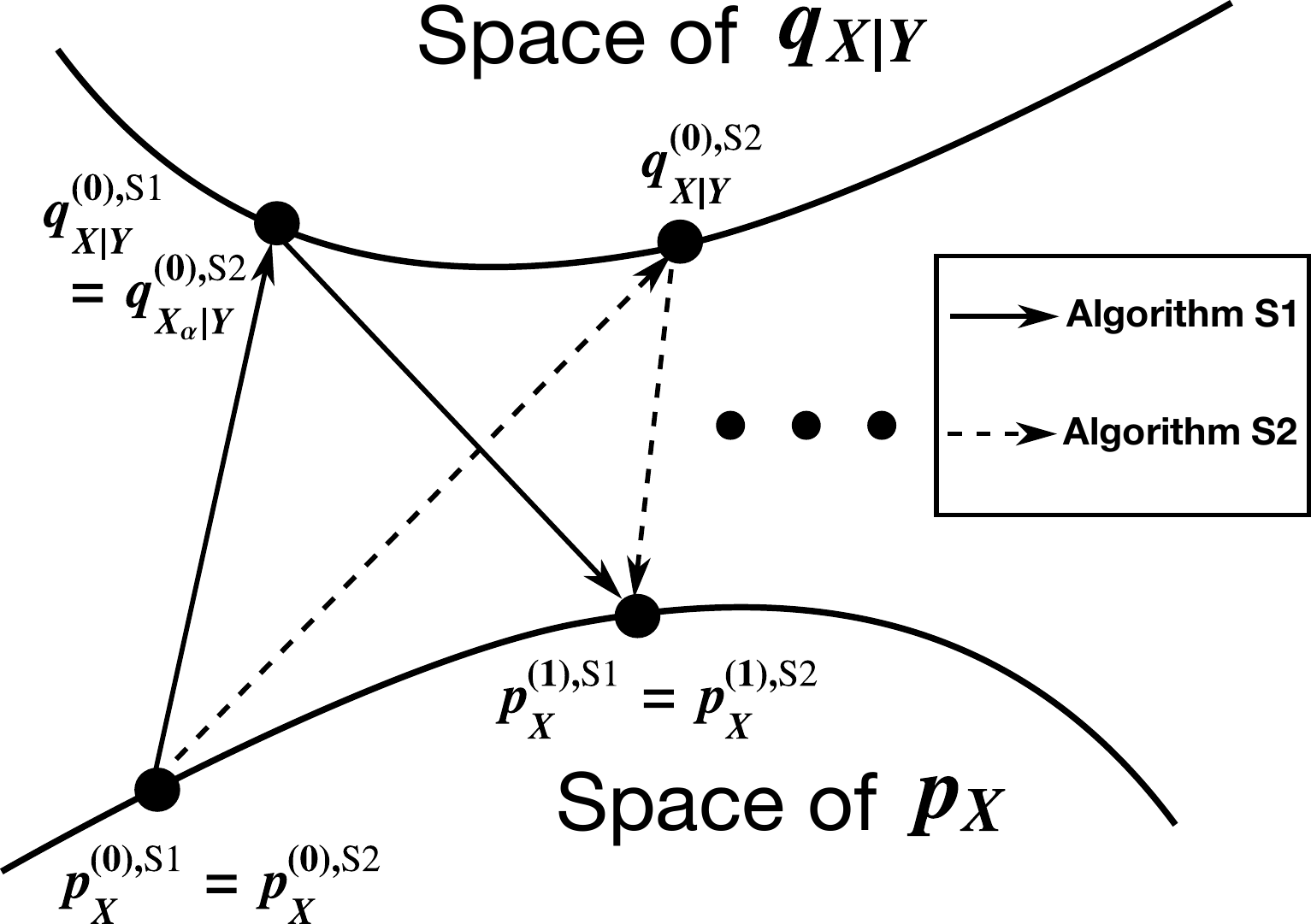}
\caption{Iterations of $p_{X}^{(k)}$ and $q_{X\mid Y}^{(k)}$ in Algorithm S1 (solid line) and Algorithm S2 (dashed line) when $p_{X}^{(0), \text{S1}} = p_{X}^{(0), \text{S2}}$.}
\label{fig:S1_S2}
\end{figure}

\begin{figure}[htbp]
\centering
\includegraphics[width=5cm, clip]{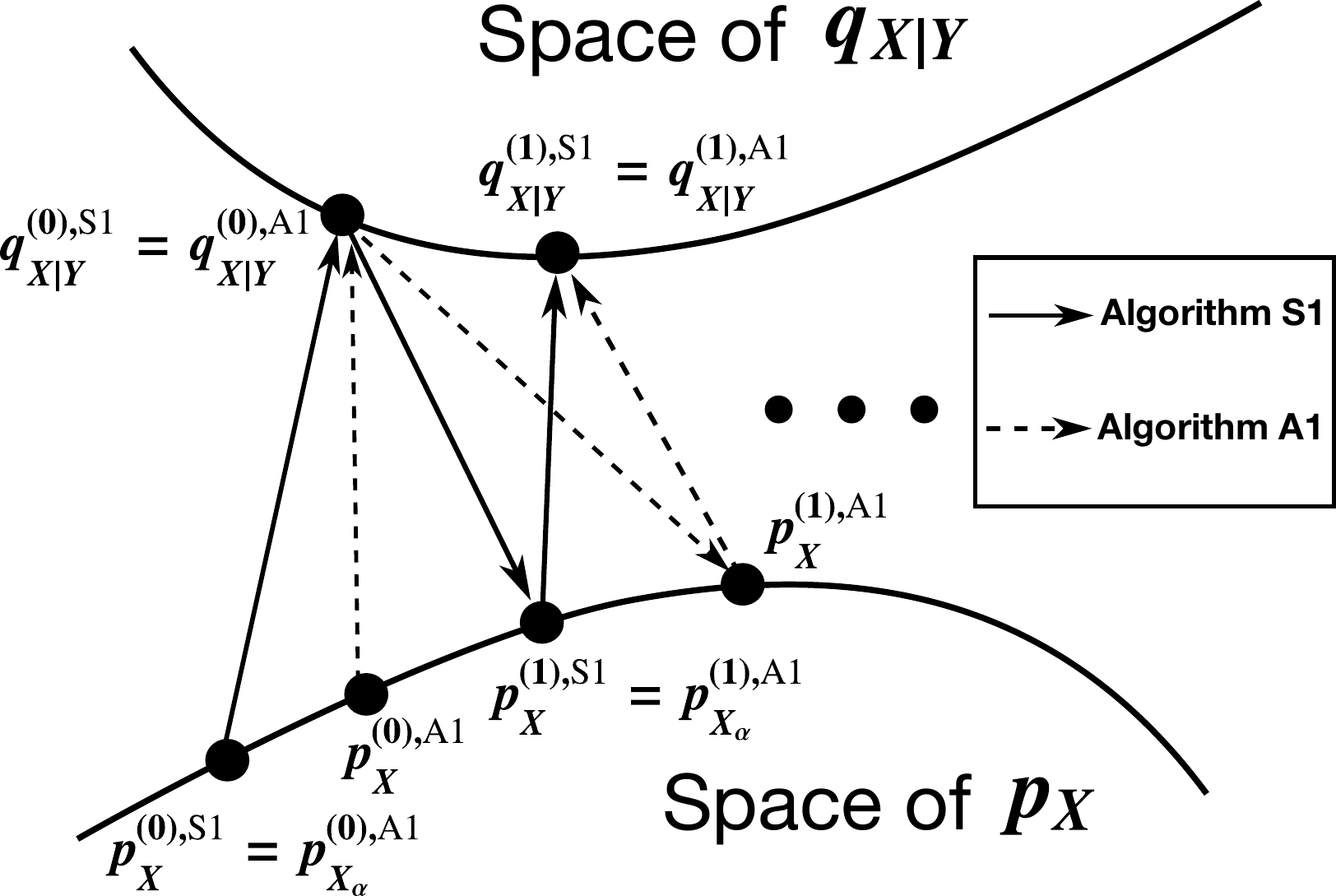}
\caption{Iterations of $p_{X}^{(k)}$ and $q_{X\mid Y}^{(k)}$ in Algorithm S1 (solid line) and Algorithm A1 (dashed line) when $p_{X}^{(0), \text{S1}} = p_{X_{\alpha}}^{(0), \text{A1}}$.}
\label{fig:S1_A1}
\end{figure}

The next corollary immediately follows from Proposition \ref{prop:properies_alpha_dist} 3) and Theorem \ref{thm:equivalence_S1_S2_A1_A2}. 
\begin{cor} \label{cor:equivalence}
Suppose that $p_{X}^{(0), \text{S1}} = p_{X}^{(0), \text{S2}} = p_{X_{\alpha}}^{(0), \text{A1}}=p_{X_{\alpha}}^{(0), \text{A2}}$. 
Then, the algorithms exhibit identical behavior. 
In particular, if $p_{X}^{(0), \text{S1}}, p_{X}^{(0), \text{S2}}, p_{X}^{(0), \text{A1}}$, and $p_{X}^{(0), \text{A2}}$ are all uniform distributions on $\mathcal{X}$, 
the condition above is satisfied.
\end{cor}

Since Algorithm S1 exhibits a global convergence property \cite[Thm 3]{1055640}, it immediately follows from Theorem \ref{thm:equivalence_S1_S2_A1_A2} that the other algorithms also exhibits this property. 
\begin{cor} \label{cor:global_convergence}
For any initial distribution $p_{X}^{(0), (\cdot)}$, algorithms S1, S2, A1, and A2  converge to the global optimum as $k \to \infty$.
\end{cor}

\subsection{Numerical Example}
In this subsection, we provide a numerical example to demonstrate Corollary \ref{cor:equivalence}. 
Let $\alpha=1.5$ and $\mathcal{X} = \mathcal{Y} = \{1, 2, 3, 4, 5\}$. 
Consider the following channel $p_{Y\mid X}$ (cited from \cite[Example 1]{nakagawa2022proof}):
\begin{align}
p_{Y\mid X} &= 
\begin{pmatrix}
0.600 & 0.100 & 0.100 & 0.100 & 0.100 \\
0.100 & 0.600 & 0.100 & 0.100 & 0.100 \\
0.231 & 0.231 & 0.066 & 0.179 & 0.292 \\
0.161 & 0.341 & 0.226 & 0.226 & 0.046 \\
0.341 & 0.161 & 0.226 & 0.046 & 0.226 \\
\end{pmatrix}, 
\end{align}
where $(i, j)$-element of channel matrix $p_{Y\mid X}$ corresponds to $p_{Y\mid X}(j | i)$. 
The initial distribution $p_{X}^{(0), (\cdot)}$ of each algorithm is a uniform distribution on $\mathcal{X}$.
Table \ref{tab:exp_f} shows the values of $F_{\alpha}^{(k, k), (\cdot)} := F_{\alpha}^{(\cdot)}(p_{X}^{(k)}, q_{X\mid Y}^{(k)})$ of each algorithm demonstrating that Corollary \ref{cor:equivalence} numerically holds.
All algorithms stopped at $k=142$, where we selected $\epsilon = 1.0\times 10^{-7}$. 

\begin{table}[htbp]
  \caption{Transition of $F^{(k, k),(\cdot)}_{\alpha}$}
  \label{tab:exp_f}
  \centering
  \resizebox{.5\textwidth}{!}{
  \begin{tabular}{@{} cccccccc @{}}
    \toprule
    $k$ & $0$  & $10$ & $50$ & $90$ & $130$ & $142$ \\ 
    \midrule
    $F_{\alpha}^{(k, k) ,\text{S1}}$ & $0.23906$ & $0.26389$ & $0.26554$ & $0.26558$ & $0.26559$ & $0.26559$ \\ 
    $F_{\alpha}^{(k, k) ,\text{S2}}$ & $0.23906$ & $0.26389$ & $0.26554$ & $0.26558$ & $0.26559$ & $0.26559$ \\ 
    $F_{\alpha}^{(k, k) ,\text{A1}}$ & $0.23906$ & $0.26389$ & $0.26554$ & $0.26558$ & $0.26559$ & $0.26559$ \\ 
    $F_{\alpha}^{(k, k) ,\text{A2}}$ & $0.23906$ & $0.26389$ & $0.26554$ & $0.26558$ & $0.26559$ & $0.26559$ \\ 
    \bottomrule
  \end{tabular}
  }
\end{table}

The obtained optimal distributions are as follows:
\begin{align}
    p_{X}^{(142),\text{S1}} &=  p_{X}^{(142),\text{S2}} \approx (0.361,  0.351,  0.115,  0.121, 0.0518), \label{eq:exp_px_s1} \\
    p_{X}^{(142),\text{A1}} &= p_{X}^{(142),\text{A2}} \approx ( 0.419,  0.401,  0.075, 0.082, 0.023) \\ 
    &\approx p_{X_{\alpha}}^{(142),\text{S1}} =  p_{X_{\alpha}}^{(142),\text{S2}}. 
\end{align}

\section{Conclusion}\label{sec:conclusion}
In this paper, we proposed novel algorithms for computing Sibson capacity $C_{\alpha}^{\text{S}}$ and 
Arimoto capacity $C_{\alpha}^{\text{A}}$ by employing the $\alpha$-tilted distribution and H\"{o}lder's inequality. 
Furthermore, we proved the equivalence of the proposed algorithms (Algorithms S2 and A2) with previous algorithms (Algorithms S1 and A1) by selecting appropriate initial distributions for these algorithms.
Using the equivalence result, we also proved that all algorithms exhibit the global convergence property. 
In a future study, we will derive an algorithm for calculating Augustin--Csisz\`{a}r capacity $C_{\alpha}^{\text{C}}$ directly. 


\appendices

\section{Proof of Proposition \ref{prop:properies_alpha_dist}} \label{proof:properies_alpha_dist}
\begin{proof}
\noindent
1)
\begin{align}
p_{\left( X_{\alpha} \right)_{\beta}}(x) &:= \frac{p_{X_{\alpha}}(x)^{\beta}}{\sum_{x}p_{X_{\alpha}}(x)^{\beta}} = \frac{\left( \frac{p_{X}(x)^{\alpha}}{\sum_{x}p_{X}(x)^{\alpha}} \right)^{\beta}}{\sum_{x}\left( \frac{p_{X}(x)^{\alpha}}{\sum_{x}p_{X}(x)^{\alpha}} \right)^{\beta}} \\
&= \frac{p_{X}(x)^{\alpha\beta}}{\sum_{x}p_{X}(x)^{\alpha\beta}} = p_{X_{\alpha\beta}}(x).
\end{align}
\noindent
2) Let $p_{X}^{*}$ and $p_{X}^{*, \alpha}$ be optimal distributions that maximize $\mathcal{F}(p_{X})$ and $\mathcal{F}(p_{X_{\alpha}})$, respectively.
Let $p_{X_{{1}/{\alpha}}}^{*}(x) := \frac{p_{X}^{*}(x)^{1/\alpha}}{\sum_{x}p_{X}^{*}(x)^{1/\alpha}}$ be the $\frac{1}{\alpha}$-tilted distribution of $p_{X}^{*}$. 
From the definition and the Proposition 1) it follows that 
\begin{align}
&\max_{p_{X}}\mathcal{F}(p_{X_{\alpha}}) = \mathcal{F}(p_{X}^{*, \alpha}) \leq \max_{p_{X}} \mathcal{F}(p_{X}) \notag \\
&= \mathcal{F}(p_{X}^{*}) = \mathcal{F}(p_{(X_{{1}/{\alpha}})_{\alpha}}^{*}) \leq \max_{p_{X}} \mathcal{F}(p_{X_{\alpha}}).
\end{align}

\noindent
3) $(\Leftarrow)$: Assume that $p_{X}(x) = 1/\abs{\mathcal{X}}, x\in \mathcal{X}$. Then 
$p_{X_{\alpha}}(x) = \frac{1/\abs{\mathcal{X}}^{\alpha}}{\sum_{x}1/\abs{\mathcal{X}}^{\alpha}} = 1/\abs{\mathcal{X}}$. 

\noindent
$(\Rightarrow)$:
Assume that $p_{X}(x) = \frac{p_{X}(x)^{\alpha}}{\sum_{x}p_{X}(x)^{\alpha}}, x\in \mathcal{X}$. This implies $p_{X}(x) = \left( \sum_{x}p_{X}(x)^{\alpha} \right)^{1/(\alpha-1)}$. 
Summing over $x\in \mathcal{X}$, we have $1=\abs{\mathcal{X}} \left( \sum_{x}p_{X}(x)^{\alpha} \right)^{1/(\alpha-1)}$. 
Therefore,  $p_{X}(x) = 1/\abs{\mathcal{X}}, x\in \mathcal{X}$. 
\end{proof}

\section{Proof of Theorem \ref{thm:novel_max_characterization}} \label{proof:novel_max_characterization}
We first review H\"{o}lder's inequality and its equality condition.
\begin{lemma}[H\"{o}lder's inequality \cite{Holder1889}] \label{lemma:Holder}
For $p>0$ and $a_{i}\geq 0, b_{i}\geq 0, i=1,\dots, n$, 
\begin{align}
\begin{cases}
\sum_{i=1}^{n} a_{i}b_{i} \leq \left( \sum_{i=1}^{n} a_{i}^{p} \right)^{1/p} \left( \sum_{i=1}^{n}b_{i}^{\frac{p}{p-1}} \right)^{1-1/p}, \quad p > 1 \\
        \sum_{i=1}^{n} a_{i}b_{i} \geq \left( \sum_{i=1}^{n} a_{i}^{p} \right)^{1/p} \left( \sum_{i=1}^{n}b_{i}^{\frac{p}{p-1}} \right)^{1-1/p}, \quad 0 < p < 1, 
\end{cases} \label{eq:Holder_ineq}
\end{align}
where the equality holds if and only if there exists a constant $c$ such that for all $i=1,\dots, n$, 
$a_{i}^{p} = cb_{i}^{\frac{p}{p-1}}$.
\end{lemma}

Using the H\"{o}lder's inequality, we prove Theorem \ref{thm:novel_max_characterization} as follows.
\begin{proof}
We only prove \eqref{eq:novel_max_characterization_Sibson_MI}. 
\eqref{eq:novel_max_characterization_Arimoto_MI} can be proved similarly. Additionally, we can provide alternative proof for Proposition \ref{prop:max_characterization_Sibson_Arimoto_MI}.

For a fixed $p_{X}$, we obtain 
\begin{align}
&\tilde{F}_{\alpha}^{\text{S}}(p_{X}, q_{X\mid Y}) \notag \\
&=\frac{\alpha}{\alpha-1}\log \sum_{y}\sum_{x}p_{X}(x)^{\frac{1}{\alpha}}p_{Y\mid X}(y\mid x)q_{X_{\alpha}\mid Y}(x\mid y)^{\frac{\alpha-1}{\alpha}} \\
&\overset{(a)}{\leq} \frac{\alpha}{\alpha-1} \log \sum_{y} \left( \sum_{x} p_{X}(x)p_{Y\mid X}(y\mid x)^{\alpha} \right)^{\frac{1}{\alpha}} \notag \\
&\qquad \qquad \qquad \qquad \qquad \times \left(\underbrace{\sum_{x} q_{X_{\alpha}\mid Y}(x\mid y)^{\frac{\alpha-1}{\alpha}\cdot \frac{\alpha}{\alpha-1}}}_{=1} \right)^{1-\frac{1}{\alpha}} \\  \label{eq:ineq}
&=\frac{\alpha}{\alpha-1} \log \sum_{y}\left( \sum_{x}p_{X}(x)p_{Y\mid X}(y\mid x)^{\alpha} \right)^{\frac{1}{\alpha}}
= I_{\alpha}^{\text{S}}(X; Y), 
\end{align}
where $(a)$ follows from H\"{o}lder's inequality\footnote{Note that for $0<\alpha<1$, $\frac{\alpha}{\alpha-1}<0$.} applied for each $y\in \mathcal{Y}$. 
The equality holds if for each $y\in \mathcal{Y}$, there exists a constant $c_{y}$ such that for  all $x\in \mathcal{X}, p_{X}(x)p_{Y\mid X}(y | x)^{\alpha} = c_{y} q_{X_{\alpha}\mid Y}(x | y)$. 
Solving this with respect to $q_{X\mid Y}(x | y)$ and using $1=\sum_{x}q_{X\mid Y}(x | y)$, we obtain $q_{X\mid Y}(x | y) = \frac{p_{X}(x)^{\frac{1}{\alpha}}p_{Y\mid X}(y | x)}{\sum_{y}p_{X}(x)^{\frac{1}{\alpha}}p_{Y\mid X}(y | x)}$. 
\end{proof}

\section{Proof of Theorem \ref{thm:novel_update_formulae}} \label{proof:novel_update_formulae}
\begin{proof}
$1)$ and $3)$ follow immediately from Theorem \ref{thm:novel_max_characterization}. 
We only prove $2)$ here. $4)$ can be proved similarly. 

For a fixed $q_{X\mid Y}$, we obtain 
\begin{align}
&\tilde{F}_{\alpha}^{\text{S}}(p_{X}, q_{X\mid Y}) \notag \\ 
&= \frac{\alpha}{\alpha-1}\sum_{x}p_{X}(x)^{\frac{1}{\alpha}}\sum_{y}p_{Y\mid X}(y\mid x)q_{X_{\alpha}\mid Y}(x\mid y)^{\frac{\alpha-1}{\alpha}} \\ 
&\overset{(a)}{\leq} \frac{\alpha}{\alpha-1}\log \left( \underbrace{\sum_{x}p_{X}(x)^{\frac{1}{\alpha}\cdot \alpha}}_{=1} \right)^{\frac{1}{\alpha}} \notag \\
&\qquad \times \left( \sum_{x} \left( \sum_{y}p_{Y\mid X}(y\mid x)q_{X_{\alpha}\mid Y}(x\mid y)^{\frac{\alpha-1}{\alpha}} \right)^{\frac{\alpha}{\alpha-1}} \right)^{1-\frac{1}{\alpha}} \\ 
&= \log \sum_{x} \left( \sum_{y}p_{Y\mid X}(y\mid x)q_{X_{\alpha}\mid Y}(x\mid y)^{\frac{\alpha-1}{\alpha}} \right)^{\frac{\alpha}{\alpha-1}}, 
\end{align}
where $(a)$ follows from  H\"{o}lder's inequality. 
The equality holds if there exists a constant $c$ such that for all $x\in \mathcal{X}, p_{X}(x) = c\left(\sum_{y}p_{Y\mid X}\left(y\mid x\right)q_{X_{\alpha}\mid Y}(x\mid y)^{\frac{\alpha-1}{\alpha}}\right)^{\frac{\alpha}{\alpha-1}}$. 
Solving this with respect to $p_{X}$ and using $1=\sum_{x}q_{X\mid Y}(x | y)$, we obtain $p_{X}(x) = \frac{\left( \sum_{y}p_{Y\mid X}(y | x) q_{X_{\alpha}\mid Y}(x | y)^{\frac{\alpha-1}{\alpha}} \right)^{\frac{\alpha}{\alpha-1}}}{\sum_{x}\left( \sum_{y}p_{Y\mid X}(y | x) q_{X_{\alpha}\mid Y}(x | y)^{\frac{\alpha-1}{\alpha}} \right)^{\frac{\alpha}{\alpha-1}}}$. 
\end{proof}

\section{Proof of Theorem \ref{thm:equivalence_S1_S2_A1_A2}} \label{proof:equivalence_S1_S2_A1_A2}
\begin{proof}
We will only prove \eqref{eq:equivalence_S1_S2_dist} and \eqref{eq:equivalence_S1_S2_value}. 
\eqref{eq:equivalence_A1_A2_dist}, \eqref{eq:equivalence_A1_A2_value}, \eqref{eq:equivalence_S1_A1_dist}, \eqref{eq:equivalence_S1_A1_value}, \eqref{eq:equivalence_S2_A2_dist}, and \eqref{eq:equivalence_S2_A2_value} can be proved similarly.

First, we will prove \eqref{eq:equivalence_S1_S2_dist} by induction on $k$. 
When $k=0$, from the assumption and the updating formula in Table \ref{tab:update_formula}, 
it immediately follows that $p_{X}^{(0), \text{S1}} = p_{X}^{(0), \text{S2}}(=: p_{X}^{(0)})$ and 
\begin{align}
q_{X\mid Y}^{(0), \text{S1}}(x\mid y) &= \frac{p_{X}^{(0)}(x)p_{Y\mid X}(y\mid x)^{\alpha}}{\sum_{x}p_{X}^{(0)}(x)p_{Y\mid X}(y\mid x)^{\alpha}}  \\ 
&= \frac{\left( \frac{p_{X}^{(0)}(x)^{\frac{1}{\alpha}}p_{Y\mid X}(y\mid x)}{\sum_{x}p_{X}^{(0)}(x)^{\frac{1}{\alpha}}p_{Y\mid X}(y\mid x)} \right)^{\alpha}}{\sum_{x}\left( \frac{p_{X}^{(0)}(x)^{\frac{1}{\alpha}}p_{Y\mid X}(y\mid x)}{\sum_{x}p_{X}^{(0)}(x)^{\frac{1}{\alpha}}p_{Y\mid X}(y\mid x)} \right)^{\alpha}} \\ 
&= \frac{q_{X\mid Y}^{(0), \text{S2}}(x\mid y)^{\alpha}}{\sum_{x}q_{X\mid Y}^{(0), \text{S2}}(x\mid y)^{\alpha}} = q_{X_{\alpha}\mid Y}^{(0), \text{S2}}(x\mid y). \label{eq:initial_induction}
\end{align}
Suppose that \eqref{eq:equivalence_S1_S2_dist} and \eqref{eq:equivalence_S1_S2_value} hold for $k$. Then, using the updating formulae in Table \ref{tab:update_formula}, we obtain  
\begin{align}
&p_{X}^{(k+1), \text{S1}}(x) 
= \frac{\left( \sum_{y}p_{Y\mid X}(y\mid x)q^{(k), \text{S1}}_{X\mid Y}(x\mid y)^{\frac{\alpha-1}{\alpha}} \right)^{\frac{\alpha}{\alpha-1}}}{\sum_{x} \left( \sum_{y}p_{Y\mid X}(y\mid x)q^{(k), \text{S1}}_{X\mid Y}(x\mid y)^{\frac{\alpha-1}{\alpha}} \right)^{\frac{\alpha}{\alpha-1}}} \\ 
&= \frac{\left( \sum_{y}p_{Y\mid X}(y\mid x)q^{(k), \text{S2}}_{X_{\alpha}\mid Y}(x\mid y)^{\frac{\alpha-1}{\alpha}} \right)^{\frac{\alpha}{\alpha-1}}}{\sum_{x} \left( \sum_{y}p_{Y\mid X}(y\mid x)q^{(k), \text{S2}}_{X_{\alpha}\mid Y}(x\mid y)^{\frac{\alpha-1}{\alpha}} \right)^{\frac{\alpha}{\alpha-1}}}
= p_{X}^{(k+1), \text{S2}}(x).
\end{align}
Equality $q_{X\mid Y}^{(k+1), \text{S1}}(x|y) = q_{X_{\alpha}\mid Y}^{(k+1), \text{S2}}(x|y)$ can be shown in a similar way as \eqref{eq:initial_induction}. 

Next, we prove \eqref{eq:equivalence_S1_S2_value}. 
\begin{align}
F_{\alpha}^{(k, k), \text{S1}} &= F_{\alpha}^{\text{S}}(p_{X}^{(k), \text{S1}}, q_{X\mid Y}^{(k), \text{S1}}) \\ 
&\overset{(a)}{=} F_{\alpha}^{\text{S}}(p_{X}^{(k), \text{S2}}, q_{X_{\alpha}\mid Y}^{(k), \text{S2}}) \\ 
&\overset{(b)}{=} \tilde{F}_{\alpha}^{\text{S}}(p_{X}^{(k), \text{S2}}, q_{X\mid Y}^{(k), \text{S2}}) = F_{\alpha}^{(k, k), \text{S2}}, 
\end{align}
where 
\begin{itemize}
\item $(a)$ follows from \eqref{eq:equivalence_S1_S2_dist}, 
\item $(b)$ follows from \eqref{eq:tilde_F_S}. 
\end{itemize}
Similarly, we can show that $F_{\alpha}^{(k+1, k), \text{S1}} = F_{\alpha}^{(k+1, k), \text{S2}}$. 
\end{proof}

\end{document}